\documentclass[final,twoside,11pt]{entics} 
\usepackage{graphicx}
\usepackage{enumitem}
\usepackage[all]{xy}
\usepackage{hyperref}
\usepackage{lmodern}
\usepackage{amssymb}
\usepackage{amsmath} 
\usepackage{stmaryrd}
\DeclareMathOperator*{\bigtimes}{\vartimes}
\usepackage{enticsmacro}

\renewcommand{\mathsf}[1]{#1}
\renewcommand{\mathcal}[1]{#1}

\newcommand{\meet}{\sqcap}
\newcommand{\consist}{\ \uparrow\ }

\newcommand{\temporaryRemoved}[1]{}
\newcommand{\lo}[1]{\operatorname{#1} \,}
\newcommand{\rmin}{\lo{min}}

\newcommand{\rmax}{\lo{max}}

\newcommand{\ttt}{\lo{tt}}
\newcommand{\ff}{\lo{f\!f}}

\newcommand{\pr}{\lo{pr}}
\newcommand{\integral}{\lo{int}}
\newcommand{\osup}{\lo{sup}}

\newcommand{\oper}[1]{\operatorname{#1}}
\newcommand{\St}{\operatorname{St}}
\newcommand{\In}{\operatorname{In}}

\newcommand{\llderiv}{\operatorname{L}}
\newcommand{\cepsilon}{\operatorname{\epsilon}}
\newcommand{\Comp}{\mathrm{Comp}}

\newcommand{\nat}{\mathbb{N}}
\newcommand{\interval}{\mathbb{I}}
\newcommand{\realLine}{\mathbb{R}}
\newcommand{\realDom}{\mathbb{IR}}
\newcommand{\dualDom}{\mathbb{DR}}
\newcommand{\Dom}{\mathbb{D}}
\newcommand{\dua}{\mathord{\hbox{\makebox[0pt][l]{\raise .6mm
                           \hbox{$\uparrow$}}$\uparrow$}}}
\newcommand{\lDeriv}{{L}}
\newcommand{\real}[1]{(#1)^\mathbf{s}}
\newcommand{\infi}[1]{(#1)^\mathbf{i}}

\newcommand{\dual}[1]{(#1)^{\mathbf{d}}}
\newcommand{\sem}[1]{\llbracket{#1} \rrbracket}
\newcommand{\bsem}[1]{\mathcal{B}\sem{#1}}
\newcommand{\esem}[1]{\mathcal{E}\sem{#1}}

\def\conf{MFPS 2023} 	
\def\myconf{mfps39}
\volume{3}			
\def\volu{3}			
\def\papno{6}			
\def\lastname{Di Gianantonio, Edalat, Gutin} 
\def\runtitle{A Language for evaluating derivatives of functionals} 
\def\mydoi{12303}
\def\copynames{P. Di Gianantonio, A. Edalat, R. Gutin}    

\begin{document}

\begin{frontmatter}
\title{A Language for Evaluating Derivatives of Functionals Using Automatic Differentiation\thanksref{ALL}}
\thanks[ALL]{This work has been partially supported by the Italian MUR project PRIN 20228KXFN2 “STENDHAL”.}   
  \author{Pietro Di Gianantonio\thanksref{a}\thanksref{aemail}}	
   \author{Abbas Edalat\thanksref{b}\thanksref{bemail}}		
   \author{Ran Gutin\thanksref{b}\thanksref{cemail}}		
   \address[a]{Univesity of Udine \\ Udine, Italy}  							
   \thanks[aemail]{Email: \href{mailto:pietro.digianantonio@uniud.it} {\texttt{\normalshape
        pietro.digianantonio@uniud.it}}} 
  \address[b]{Imperial College London \\ London, UK} 
  \thanks[bemail]{Email:  \href{mailto:a.edalat@imperial.ac.uk} {\texttt{\normalshape
        a.edalat@imperial.ac.uk}}}
   \thanks[cemail]{Email:  \href{mailto:r.gutin20@imperial.ac.uk} {\texttt{\normalshape
        r.gutin20@imperial.ac.uk}}}

\begin{abstract}
We present a simple functional programming language, called Dual PCF, that implements forward mode automatic differentiation using dual numbers in the framework of exact real number computation.  The main new feature of this language is the ability to evaluate correctly up to the precision specified by the user --  in a simple and direct way -- the directional derivative of functionals as well as first order functions. In contrast to other comparable languages, Dual PCF also includes the recursive operator for defining functions and functionals. We provide a wide range of examples of Lipschitz functions and functionals that can be defined in Dual PCF.  We use domain theory both to give a denotational semantics to the language and to prove the correctness of the new derivative operator using logical relations.   To be able to differentiate functionals---including on function spaces equipped with their compact-open topology that do not admit a norm---we develop a domain-theoretic directional derivative that is Scott continuous and extends Clarke's subgradient of real-valued locally Lipschitz maps on Banach spaces to real-valued continuous maps on Hausdorff topological vector spaces. Finally, we show that  we can express arbitrary computable linear functionals in Dual PCF.
\end{abstract}

\end{frontmatter}

\section{Introduction}

In this paper, we describe a language for performing automatic differentiation on a wide set of functions, including higher-order functions and non-differentiable Lipschitz functions. This language combines the dual numbers---i.e., the algebra consisting of numbers of the form $a + b \varepsilon$ where $a$ and $b$ are real and $\varepsilon^2 = 0$ --- with a domain-theoretic directional derivative which we introduce in this paper. The domain-theoretic directional derivative can properly and correctly handle higher-order functions and Lipschitz functions like the absolute value function or ReLU used in machine learning. The dual numbers are used to incorporate automatic differentiation into the language in a straightforward manner.

Due to its wide range of applications, automatic differentiation has been an active theoretical and practical area of research in recent years~\cite{griewank2008evaluating}. While there is a large body of work in the subject, automatic differentiation is not always implemented in a sound and rigorous way.
A standard example is the if-then-else constructor, ubiquitous in numerical programs, 
which gives incorrect results when evaluated in automatic differentiation software \cite{abadiPlotkin19,MazzaPagani21}. 

In this work,  we develop a language, based on dual numbers and called Dual PCF, which has a rich set of definable functions. This set includes locally Lipschitz functions such as the absolute value function and ReLU that are not differentiable everywhere but are widely used in applications and have a set-valued generalised derivative called the Clarke subgradient, which generalises the classical gradient ~\cite{clarke1990optimization}.

Several formal calculi, containing a primitive for evaluating the derivative of functions, have been proposed in the literature. 
All together, four main features distinguish our calculus from these other work: developing an exact computation framework using domain theory, having a recursive operator in the language, the possibility to deal with functions that are not infinitely differentiable, and the ability to compute the derivative of functionals on real-valued functions. None of the existing languages accommodate all these four main features together.

The recursive operator allows the recursive definition of functions on real numbers. Most of the programming languages in the literature, instead, assume the existence of a sufficiently rich set of basic functions over the real numbers and construct the other functions by composition, $\lambda$-abstraction, and test operator \cite{AlvarezOng20,PicalloOng21,HuotSV22,KPK22,ShermanMC21}. Since the recursive operator is missing, these languages are not Turing complete from the point of view of real number computation. 
In \cite{KPK22}, it is claimed that the approach used in~\cite{MazzaPagani21}, together with other techniques, can be used to solve the problem of dealing with the recursive definition. However, the actual solution is left as future work. 

To accommodate the recursive operator, the denotational semantics   
 needs to describe partial elements, either in the form of partially defined functions over the real line as in \cite{abadiPlotkin19,MazzaPagani21} or as partial real numbers, as in the present work. The use of partial real numbers unifies our treatment of the derivative operator to the field of exact real number computation. Other approaches assume that real numbers form a basic type where the arithmetic operations can be computed in constant time, giving the exact result. Therefore, the computation on reals is idealised, and the problem of the infinitary nature of real numbers is wholly avoided. 
  
  A second aspect by which our work differs from existing work is that we also consider functions that are Lipschitz but not differentiable, such as the absolute value function, which plays a key role in all applications. The presence of non-differentiable functions makes the repeated application of derivative operator more complex.  The calculi in~\cite{AlvarezOng20,EhrhardR03} assume that all functions are infinitely differentiable; as a consequence, they cannot accommodate, in a coherent way, functions like absolute value or min (evaluating the minimum of two real numbers) that are Lipschitz but not everywhere differentiable.
The formal languages developed in~\cite{abadiPlotkin19,MazzaPagani21} accommodate a larger set of functions including the if-then-else constructor, but define a restricted domain for the input-values and only ensure the correct evaluation of the derivative in the restricted domain. 


However, the main feature of our calculus is the ability to compute the derivative of functionals on real-valued functions. We therefore extend the mechanism of automatic differentiation to provide the directional derivative of functionals.  To this end, we have developed a domain-theoretic Scott continuous directional derivative for real-valued functions on Hausdorff topological vector spaces that extends the Clarke subgradient to functionals defined on function spaces where the Scott topology does not admit a norm. In addition, we are able to reduce the problem of definability of linear functionals in the language to the definability of first-order functions.

 Automatic differentiation on functionals has also been considered in~\cite{frechet,HuotDV20,KPK22,ShermanMC21}. In \cite{frechet}, an approach quite different from ours is used: functions on reals are represented through a base of Chebyshev polynomials, and thereby the directional derivative of a functional is evaluated. In \cite{HuotDV20,ShermanMC21} no recursive definition of functions on reals is possible and we have already commented on \cite{KPK22}.  

Because both recursion on real numbers and non-differentiable Lipschitz maps are included in our language, we cannot use the approach adopted by \cite{HuotSV22,ShermanMC21} in the denotational semantics, which creates a serious challenge for defining a denotational semantics as we undertake in this work. Here, we  establish the correctness of the results given by automatic differentiation of higher-order functions and relate them to mathematical notions of derivatives. 

 The new domain-theoretic notion of directional derivative developed in the paper computes the support function of the well-established mathematical concept of the generalized Lebourg's subgradient which reduces, when this subgradient is a singleton, to the Gateaux derivative of real-valued maps on topological vectors spaces. In contrast, in \cite{ShermanMC21,HuotSV22} the directional derivatives developed using category theory are not related to any established mathematical notions of differentiation on topological vector spaces.
In this respect, in \cite{HuotSV22},   Section 7.2., which has the title ``Canonical derivatives of higher order functions?", concludes with the following remark:
``We hope that an exploration of such techniques might lead to an appropriate notion of computable derivative, even for higher order functions."
Thus, the authors leave the problem of defining a canonical notion of derivative for higher order functions as an open problem. We claim that our approach satisfies at least some of the requirements for a canonical notion of directional derivative since we have shown the coincidence between the standard mathematical notion of derivative and a constructive notion of derivative obtained through domain theory, which, in turn, induces an effective way to differentiate functionals.

Dual PCF is simply a functional programming language with an extra basic type of dual numbers. 
In fact, dual numbers are at the basis of a standard approach to automatic differentiation, in which one evaluates the derivative of a function at a given point and along one variable.  In our language, however, we show also that dual numbers can be used to obtain the directional derivative of a function in several variables, and the directional derivative of a functional. An important contribution is a formal proof of correctness for the computation of the directional derivative inside the language. 

We also note that there exists a trick, well-known in functional analysis, to reduce the problem of evaluating the derivative of a functional along a given direction to the problem of evaluating the derivative of a first order function.  Given a functional $F : (\realLine \to \realLine) \to \realLine$, the derivative of $F$ at $f : \realLine \to \realLine$ in the direction $g : \realLine \to \realLine$ can be reduced to evaluating the derivative of $\lambda x . F(  \lambda y . f(y) + x \cdot g(y))$.  However, this technique has its own problems. Without an analysis of how functionals are described in automatic differentiation, there is no guarantee that the above technique will evaluate the derivative correctly. More specifically, our language contains higher-order primitives such as integration or supremum, and it is necessary to check that automatic differentiation is correctly implemented on these primitives. Furthermore, it is useful to be able to evaluate the derivative of a functional directly from the functional itself, without having to plunge the functional into another expression just to extract the derivative.  

The wide range of applications of differentiation of functionals, some of which we elaborate in the paper, include:
 calculus of variations~\cite{variations}, numerical solution of differential equations using Newton's method over function spaces~\cite{frechet}, optimal control theory which features some non-differentiable functions such as the absolute value function~\cite{clarke}, physical applications in analytical mechanics, Lagrangian and Hamiltonian mechanics~ \cite{sicm,analytical} and in quantum chemistry~ \cite{kineticenergy}.
 
As one of our key applications, we show how in Dual PCF one can solve initial value problems in the theory of ordinary differential equations~\cite{hirsch1974differential}. Using the constant for integration in the language, we are able to construct Picard's functional for solving an initial value problem in the framework of domain theory~\cite{edalat2007domain}. To our knowledge, this is the first time a (theoretical) functional programming language can undertake this task.  

\subsection{Related works}

Several simple calculi with a derivative operator as a primitive have been presented in the literature.  A few of these do not aim to implement automatic differentiation \cite{EhrhardR03,DigEdalat13}.
Forward-mode automatic differentiation is realized in \cite{Manzyuk12,ManzyukPRRS19,ShermanMC21}.

In recent years, driven by the application of reverse-mode automatic differentiation in deep learning, a series of formal calculi have been proposed, some of them based on reverse-mode automatic differentiation \cite{BrunelMP20,Elliott18,MakOng20}, others implementing both forward-mode and reverse-mode automatic differentiation \cite{MazzaPagani21}. In these works, the correctness of automatic differentiation is proved in the context of first order languages \cite{abadiPlotkin19}, or higher order functional languages \cite{AlvarezOng20,BrunelMP20,Elliott18,KPK22,MakOng20,MazzaPagani21}.
Calculi that use logical relations to prove the correctness of derivative evaluation, a key aspect of our work, are provided in \cite{BartheCDG20,DigEdalat13}.
In several of these works, the semantics of differentiation is given using a categorical setting, \cite{AlvarezOng20,HuotDV20,MakOng20}, in contrast to using the more concrete space of real numbers or its extension in domain theory. 
With the exception of \cite{DigEdalat13,ShermanMC21}, none of these formalisms uses the notion of Clarke subgradient. 
The main new result over our previous work \cite{DigEdalat13} is the extension of the notion of derivative to functionals. This extension involves, among other things, transitioning from the Clarke's subgradient for Banach space to generalized Lebourg's subgradient for topological vector space. Also, we consider the use of dual numbers to provide an alternative operational semantics. Additionally, we show that all computable linear functionals can be expressed in Dual PCF.

The paper is organised as follows. In the rest of this section, we recall the elementary facts about dual numbers, followed by the domain-theoretic notions and results required in this paper. In Section~\ref{directional-l-derivative}, we develop the domain-theoretic generalization of Clarke's subgradient for Hausdorff topological vector spaces. In Section~\ref{star-op}, we introduce the domain of dual numbers. In Section~\ref{language}, we present the syntax of Dual PCF, with its denotational and operational semantics, and prove adequacy. In Section~\ref{examples}, a wide range of examples of functions and functionals definable in Dual PCF are presented. In section~\ref{consistency-section}, we define the notion of local consistency between the real and dual parts of a function on the dual domain. In Section~\ref{correctness}, we show that the semantic interpretations of the functions definable in Dual PCF are locally consistent. Finally, in Section~\ref{lin-func}, we prove that all linear functionals on real functions that vanish at infinity are definable in Dual PCF.

All proofs for the results are provided in the full version of this paper~\cite{DEG22}.

\subsection{Dual-number preliminaries}
 
The dual numbers are one of only three 2-dimensional ``number systems'' that extend the real numbers (\cite{kantor89}). A dual number is an expression of the form $a+\varepsilon b$, where \(a\) and  \(b\)  are reals, and \(\varepsilon\) is a new type of imaginary number with \(\varepsilon^2=0\).

The \(\varepsilon\) can also be thought of as representing some very small number (\cite{bell_2008}). The number is not \emph{so small} as to be zero, but it is small enough that its product with itself is zero. This leads to an intuitive picture where the dual numbers represent a one-dimensional number line where each number on that line is surrounded by a set of numbers which are infinitely close to it. We call the term \(a\) in \(a + b\varepsilon\) the \emph{standard part}, and the term \(b\) the \emph{infinitesimal part}. This terminology is consistent with that of nonstandard analysis.

To appreciate the properties of dual numbers, let \(f(x)\) be some polynomial. It is easy to check that
\begin{equation} \label{consistent}
    f(x + \varepsilon x') = f(x) + \varepsilon x' \cdot f'(x),
\end{equation} 
where $f'$ is the derivative of $f$. On a purely algebraic level, the above equation shows that the dual numbers are able to ``accidentally'' differentiate an arbitrary function. This feature of the dual numbers can be used to ``induce'' a computer language into computing the derivative of a subroutine, essentially by exploiting operator overloading (\cite{Corliss93operatoroverloading}). The above presentation can be seen as an alternative formulation of ``forward-mode automatic differentiation'' (\cite{hitchhiker}), a method that is notable for not introducing any numerical approximations, and also for avoiding the exponential overhead of naive symbolic differentiation. 

\section{Scott topology and directional derivative on topological vector spaces}\label{directional-l-derivative}

We first present the elements of domain theory and topology required here; see~\cite{abramsky1995domain} and~\cite{gierz2003continuous} for basic references to domain theory. We denote the closure and interior of a subset $S$ of a topological space by $\overline{S}$ and $S^\circ$ respectively. For a map $f:X\to Y$ of topological spaces $X$ and $Y$, denote the image of any subset $S\subset X$ by $f[S]$. The compact-open topology of the function space $(X\to Y)$ has sub-basic open sets of the form $(C,O)=\{f:f[C]\subset O\}$, with $C$ compact and $O$ open.  If $X$ and $Y$ are metric spaces, $f$ is {\em locally Lipschitz} if for any $x\in X$ there exist an open neighbourhood $O$ of $x$ and  $k\geq 0$ such that $d_Y(f(x_1),f(x_2))\leq kd_X(x_1,x_2)$ for all $x_1,x_2\in O$. If $I\subset \realLine$ is a non-empty compact real interval, we write $I=[I^-,I^+]$.

A directed complete partial order $D$ is a partial order in which every directed set $A\subset D$ has a lub (least upper bound) or supremum $\bigsqcup A$. The way-below relation $\ll$ in a dcpo $(D,\sqsubseteq)$ is defined by $x\ll y$ if whenever there is a directed subset $A\subset D$ with $y\sqsubseteq \bigsqcup A$, then there exists $a\in A$ with $x\sqsubseteq a$. A subset $B\subset D$ is a basis if for all $y\in D$ the set $\{x\in B:x\ll y\}$ is directed with lub $y$. By a {\em domain} we mean a dcpo with a basis. Domains are also called continuous dcpo's. If $D$ has a countable base then it is called a countably based domain.  In a domain $D$ with basis $B$, we have the interpolation property: the relation $x\ll y$, for $x,y\in D$, implies there exists $z\in B$ with $x\ll z\ll y$. 

A subset $A\subset D$ is {\em bounded} if there exists $d\in D$ such that for all $x\in A$ we have $x\sqsubseteq d$. If a pair of elements $d_1,d_2\in D$ is bounded above (consistent), we write $d_1\uparrow d_2$ and refer to the predicate $\uparrow$ as the {\em consistency relation}. If any bounded subset of $D$ has a lub then $D$ is called bounded complete. In particular a bounded complete domain has a bottom element $\bot$ that is the lub of the empty subset. A bounded complete domain $D$ has the property that any non-empty subset $S\subset D$ has an infimum or greatest lower bound $\bigsqcap S$. All domains in this paper are bounded complete and countably based. 

The set of non-empty compact intervals of the real line ordered by reverse inclusion and augmented with the whole real line as bottom is the prototype bounded complete domain for real numbers denoted by $\realDom$, in which $I\ll J$ iff $J \subset I^\circ$. It has a basis consisting of all intervals with rational endpoints.  For two non-empty compact intervals $I$ and $J$, their infimum $I\sqcap J$ is the convex closure of $I\cup J$. The Scott topology on a domain $D$ with basis $B$ has sub-basic open sets of the form $\dua b:=\{x\in D:b\ll x\}$ for any $b\in B$. The upper set of an element $x\in D$ is given by $\uparrow x=\{y\in D:x\sqsubseteq y\}$. 

The lattice of Scott open sets of a bounded complete domain is continuous. The basic Scott open sets for $\realDom$ are of the form $\{J\in \realDom:J\subset I^\circ\}$ for any $I\in \realDom$. The maximal elements of $\realDom$ are the singletons $\{x\}$ for $x\in \realLine$ which we identify with real numbers, i.e., we write $\realLine\subset \realDom$, as the mapping $x\mapsto \{x\}$ is a topological embedding when $\realLine$ is equipped with its Euclidean topology and $\realDom$ with its Scott topology. Similarly, $\interval [a,b]$ is the domain of non-empty compact intervals of $[a,b]$ ordered with reverse inclusion. 

If $X$ is any topological space with some open set $O\subset X$ and $d\in D$ lies in the domain $D$, then the single-step function $d\chi_O:X\to D$, defined by $d\chi_O(x)=d$ if $x\in O$ and $\bot$ otherwise, is a Scott continuous function. The partial order on $D$ induces by point-wise extension a partial order on continuous functions of type $X\to D$ with $f\sqsubseteq g$ if $f(x)\sqsubseteq g(x)$ for all $x\in X$.  For any two bounded complete domains $D$ and $E$, the function space $(D\to E)$ consisting of Scott continuous functions from $D$ to $E$ with the extensional order is a bounded complete domain with a basis consisting of lubs of bounded and finite  families of single-step functions. If the lattice $\Omega X$ of open sets of $X$ is a domain and if $D$ is a bounded complete domain then for any continuous function $f:X\to D$ we have $d\chi_O\ll f$ iff $O\ll_{\Omega X} f^{-1}(\dua d)$.~\cite[Proposition II-4.20(iv)]{gierz2003continuous}. 

\begin{proposition}\label{envelope} \cite[Exercise II-3.19]{gierz2003continuous}
 If $h:A\subset Y\to D$ is any map from a dense subset $A$ of $Y$ into a bounded complete domain $D$, then its envelope 
 \[h^\star:Y\to D\]
 given by $h^\star(x)=\bigsqcup\{\bigsqcap h[O]: x\in O \mbox{ open}\}$ is a continuous map with $h^\star(x)\sqsubseteq h(x)$, and in addition $h^\star(x)=h(x)$ if $h$ is continuous at $x\in A$. Moreover, $h^\star$ is the greatest continuous function $p:Y\to D$ with $p(x)\sqsubseteq h(x)$ for all $x\in Y$.
\end{proposition}
Since $\realLine\subset \realDom$ is dense, any continuous map $f:\realLine\to \realLine\subset \realDom$, considered as a continuous map $f:\realLine\to \realDom$, has a maximal extension $f^\star:\realDom\to\realDom$ given by $f^\star(x)=f[x]$. We also have the following.
\begin{proposition}\label{function-sp-ext}
The set of functions \[\{f^\star:\realDom\to \realDom:f\in (\realLine\to \realLine)\}\] is dense in $(\realDom\to \realDom)$ with respect to the Scott topology.
  \end{proposition}

  In this paper, we construct a language for differentiation of functions of first-order and functionals of second-order. The input and output of these functions and functionals will be given by bounded complete domains whose set of maximal elements consists of real numbers, in the case of functions, or contains the continuous real-valued functions, in the case of functionals. We will first study the topological properties of the space of first-order real-valued functions in this section.
\subsection{Scott topologies on function spaces}

The space of maximal elements of a bounded complete domain, equipped with its relative Scott topology, is Hausdorff~\cite{kamimura1984total} and is in fact a complete separable metrisable space~\cite{lawson1997spaces}. Moreover, the domains encountered in this paper are constructed from $\realDom$ by using Cartesian product and function space construction, and therefore they inherit the operations for addition and scalar multiplication from interval arithmetic operations on $\realDom$; e.g., for $f,g:\realDom\to\realDom$ we have $f+g:\realDom\to\realDom$, with $(f+g)(x)=f(x)+g(x)$, where $f(x)+g(x)$ is the sum of two intervals $f(x)$ and $g(x)$.

As a major example that we work with in this paper, consider the function space $(\realLine\to \realLine)$ of all continuous real-valued functions on the real line. This function space is a real vector space with the usual operations of addition of functions and multiplication by real numbers. It is in one to one correspondence with the subset of functions in the maximal elements of the bounded complete domain $(\realDom \to \realDom)$ consisting of the maximal extensions of continuous real-valued functions. This subset of maximal elements inherits the relative Scott topology from the domain, but does not admit a norm (see Proposition~\ref{not-normable} below).

The compact-open topology on the function space $(Y\to Z)$, the collection of all continuous functions between topological spaces $Y$ and $Z$, has sub-basic open sets of the form
\[{ U}(C,O):=\{f:Y\to Z: f[C]\subset O\},\]
for any compact subset $C\subset Y$ and open $O\subset Z$~\cite{dugundji1966topology}. There is a simple characterisation of the compact-open topology for the function space $(\realLine \to \realLine)$ or $([0,1] \to \realLine)$. 
\begin{lemma}\label{comp-open}
The compact-open topology on $(\realLine\to\realLine)$ is generated by the sub-basis consisting of subsets of the form $U(\overline{O_1},O_2)$ where $O_1$ and $O_2$ are open intervals with compact closures.
\end{lemma}

Let $\operatorname{Max}(D)$ denote the set of maximal points of a bounded complete domain $D$ and
let $(\_)^\star:(\realLine\to\realLine)\to (\realDom\to \realDom)$ be the maximal extension (envelope) operator given in Proposition~\ref{envelope}, where $(\realLine\to\realLine)$ is equipped with the compact-open topology.
\begin{proposition}\label{top-embed}
 The map $(\_)^\star$ is a topological embedding, i.e. it is injective, continuous and is an open map onto its image in $ \operatorname{Max}(\realDom\to \realDom)$ with respect to the relative Scott topology on $\operatorname{Max}(\realDom\to \realDom)$.
\end{proposition}

The function space $(\realLine\to \realLine)$, equipped with the compact-open topology, is an example of a Hausdorff topological vector space. Recall that a Hausdorff topological vector space is a vector space with a Hausdorff topology with respect to which addition of vectors and scalar multiplication are continuous operations. 

\begin{proposition}\label{not-normable}
 
The function space $(\realLine\to \realLine)$ with the compact-open topology does not admit a norm.
\end{proposition}

We make two additional remarks here. The sup norm topology on the function space  $([0,1]\to\realDom)$ coincides with the compact-open topology as is easy to check. However, the compact-open topology on the function space $(\realLine\to_{{b}} \realLine)$, the set of bounded continuous maps, is strictly weaker than the sup norm topology~\cite[p. 284]{dugundji1966topology}. The same is true for the space $C_0(\realLine):=(\realLine\to_0 \realLine)$, the set of continuous maps of type $\realLine\to \realLine$ that vanish at infinity.

Therefore, for computational reasons, we will work with Hausdorff topological vector spaces which are more general than normed vector spaces and give a unifying framework for function spaces as well as the more basic finite dimensional Euclidean spaces we consider in this paper. 

Finally, we have the following result which follows from Proposition~\ref{function-sp-ext}. Consider the function space $(\realLine\to\realLine)\to \realLine$ with its compact-open topology.
\begin{corollary}\label{ext-functionals}
 Any continuous functional $F: (\realLine\to\realLine)\to \realLine$ has a continuous extension $F^\star:(\realDom\to\realDom)\to \realDom$ obtained by first extending $F$ by $F(f^\star)=F(f)$ for any $f\in (\realLine\to \realLine)$, which defines $F$ on a dense subset of $(\realDom\to\realDom)$.
  
\end{corollary}

Let $f:X\to \realLine$ be a real-valued map on a Hausdorff topological vector space $X$ (i.e., a vector space with a Hausdorff topology with respect to which addition and scalar multiplication are continuous operations). We have the following new notion of directional derivative, which generalises Clarke's generalised directional derivative to real-valued maps on topological vector spaces.
\begin{definition} \label{directionalDerivative}
  The {\em domain-theoretic directional derivative} $Lf:X\times X\to \realDom$  of $f:X\to \realLine$ at $x\in X$ in the direction $x'\in X$ is defined as
  \[Lf(x,x'):=\left[\liminf_{\stackrel{y\to x,z\to x'}{r\to 0^+} }\frac{f(y+rz)-f(y)}{r},\limsup_{\stackrel{y\to x,z\to x'}{r\to 0^+} }\frac{f(y+rz)-f(y)}{r}\right],\]
  if both endpoints above are real numbers; otherwise define $Lf(x,x')=(-\infty,\infty)=\bot$. We say $Lf(x,x')$ is {\em bounded} if it is a compact (i.e., non-bottom) interval. 
    \end{definition}
In the full version of the paper~\cite{DEG22}, we show that $Lf$ is Scott continuous, extends Clarke's notion of generalised directional derivative on a Banach space (i.e., a complete normed vector space)~\cite{clarke1990optimization} and, like the latter, satisfies a weaker calculus (with equality replaced by $\sqsubseteq$) compared to the classical derivative; this calculus includes a weaker chain rule for composition of two functions~\cite[section 2.2]{DEG22}. 

 We will work with a well-behaved family of so-called locally Lipschitzian maps on Hausdorff topological vector spaces defined as follows. Say $Lf:X\times X\to \realDom$ is {\em bounded} in the open set $O\times O'\subset X\times X$ if there exists a compact interval $C\in \realDom$ such that $ C\sqsubseteq Lf(x,x')$ for all $(x,x')\in O\times O'$.
   We say $f:X\to \realLine$ is {\em locally Lipschitzian} at $x\in X$ if there exists an open neighbourhood $O\subset X$ of $x$ and an open neighbourhood $O'$ of the origin such that the directional derivative $Lf(-,-):X\times X\to \realDom$ is bounded for $(x,x')\in O\times O'$. For a Banach space the two notions of locally Lipschitzian map and locally Lipschitz map coincide~\cite[section 2.2]{DEG22}. If $f$ is locally Lipschitzian, then for each $x\in X$, there exists, as in the construction by Lebourg in~\cite{lebourg1979generic},  a non-empty weak* compact subset, denoted $\partial f(x)$, of the dual space of $X$, such that for all $x'\in X$, we have $Lf(x,x')=\{Ax':A\in \partial f(x)\} $, which is a compact real interval~\cite[Theorem 1(i)]{DEG22}. If $X$ is a Banach space, $\partial f(x)$ coincides with Clarke's subgradient~\cite{clarke1990optimization}.


\section{Domain of dual number intervals} \label{star-op}


This section introduces a hierarchy of continuous domains for dual numbers, $\Dom_\tau $,  and defines a family of mapping, $\dual{-}_\tau$,  that embeds the spaces of functions on real numbers and those of functionals on functions into these domains.  The domain-theoretic directional derivative of a function $f$ on reals defines the \emph{infinitesimal} part of $\dual{f}_\tau$.  On the other hand, we define two families of maps, $\real{-}_\tau$ and $\infi{-}_\tau$,  that \emph{extract}, from a \emph{total} function $g$ on $\Dom_\tau$, the function $f$ on real numbers that $g$ represents as well an infinitesimal perturbation of $f$ that is used in the computation of directional derivatives.

The maps $\dual{\_}_\tau$, $\real{\_}_\tau$, and $\infi{\_}_\tau$ are given on a limited hierarchy of types $\tau$ formally defined by the following. 
We define \emph{first-order function types} to be the types having the form $\delta \to ( \ldots \to (\delta \to \delta))$, and \emph{second-order function types} to be the types having the form 
$\tau_1 \to ( \ldots \to (\tau_n \to \delta ) \ldots )$ with $\tau_1,\ldots,\tau_n$ either a first-order function type or equal to $\delta$.  Notice that, by the above definition, a first-order function type is always a second-order function type. 
We define a \emph{first-order function} to be a function $f$ having first-order function type. We define a \emph{second-order function}, or a \emph{functional}, to be a function $F$ having a strictly second-order function type. By uncurrying, a first-order function $f$ can be seen to take $n$ dual values and return a dual value.  

The domain for dual number, $\dualDom$, is the domain $\realDom \times \realDom$. The first component is the standard part of a dual number (albeit interval valued instead of single-valued) and the second component is the infinitesimal part.
The domains and codomains of the maps, $\dual{\_}_\tau$, $\real{\_}_\tau$, and $\infi{\_}_\tau$ are defined by the following.

\begin{definition}~\label{def-sec-order}
By induction on the structure of a \emph{second-order function type} $\tau$, we define the following families of topological spaces:
\begin{itemize}
    \item $\realLine_{\delta} = \realLine$ and 
    $\realLine_{\tau_1 \to \tau_2}$  is the topological vector space of continuous functions from $\realLine_{\tau_1}$ to $\realLine_{\tau_2}$ with the compact-open topology;
\item $\realLine^p_{\delta} = \realDom$ and 
$\realLine^p_{\tau_1 \to \tau_2}$  is the topological space of continuous 
functions from $\realLine_{\tau_1}$ to $\realLine^p_{\tau_2}$ with the Scott topology; 
\item $\Dom_{\delta} =  \dualDom$ 
and $\Dom_{\tau_1 \to \tau_2}$ is the bounded complete domain 
$(\Dom_{\tau_1} \to \Dom_{\tau_2})$;
\item $\Dom^t_\delta$ is the subset of $\dualDom$ containing elements whose standard part is total, i.e., maximal;
 $\Dom^t_{\tau_1 \to \tau_2}$ is the subset of $\Dom_{\tau_1 \to \tau_2}$ containing functions mapping all elements in $\Dom^t_{\tau_1 \to \tau_2}$ to elements in $\Dom^t_{\tau_2}$. Functions in $\Dom^t_{\tau}$ are called \emph{standard maximal preserving}.
\end{itemize}
\end{definition}

Notice that,  equipped with the compact-open topology, the topological vector space $\realLine_{\tau_1 \to \ldots \to \tau_n \to \delta}$ is homeomorphic to  the space $(\realLine_{\tau_1} \times \ldots \times \realLine_{\tau_n}) \to \realLine_{\delta}$~\cite[p. 261]{dugundji1966topology}. Moreover, the space of maximal elements of the listed bounded complete domains, in Definition~\ref{def-sec-order}, equipped with their relative Scott topology will be Polish, and hence Hausdorff topological vector spaces (\cite{lawson1997spaces}) and thus the results at the end of Section~\ref{directional-l-derivative}, and proved in~\cite[sections 2.1, 2.2]{DEG22}, can be used.

\begin{definition}
The functions:
$\dual{\_}_{\tau} : \realLine_\tau \to \Dom^t_\tau$, 
$\real{\_}_{\tau} 
: \Dom^t_\tau \to \realLine_\tau$,   
$\infi{\_}_{\tau} : \Dom^t_\tau \to \realLine^p_\tau$ are inductively defined on the type $\tau$ by:
\[\begin{array}{cccc}
\dual{x}_{\delta} :=  x + \varepsilon 0 &
\real{x + \varepsilon x'}_{\delta} :=  x &
\infi{x + \varepsilon x'}_\delta =  x' &\mbox{and}
\end{array}\]
\[\begin{array}{l}
\dual{f}_{\vec{\tau} \to \delta} = h^\star  
\mbox{ with } h : \Dom^t_{\tau_1} \to \ldots \to \Dom^t_{\tau_n} \to \dualDom \mbox{ defined by:} \\ 
h(\vec{d}) = f (\overrightarrow{\real{d}})) + 
\varepsilon ((\lDeriv f)^\star (\overrightarrow{\real{d}}, \overrightarrow{\infi{d}}))  \\
\real{g}_{\vec{\tau} \to \delta}(\vec{x}) = \real{g (\overrightarrow{\dual{x}})}_\delta \hspace{2em}
\infi{g}_{\vec{\tau} \to \delta}(\vec{x}) = \infi{g (\overrightarrow{\dual{x}})}_\delta, 
\end{array}\]
where  $(\_)^\star$ is the envelope operator of Proposition~\ref{envelope}.  
In defining $h$, we need to use the envelope $(\lDeriv f)^\star$ because ${\infi{d_i}}$ can be a  partial real number, or a function returning partial real numbers. 
Since the topological space $\Dom^t_\delta$ is dense in $\Dom_\delta$, by an obvious generalization of Proposition~\ref{function-sp-ext}, for any first order type $\tau$,  $\Dom^t_\tau$ is dense in $\Dom_\tau$; it follows by Proposition~\ref{envelope} 
that the envelope $h^\star$ exists for both first order and second order functions.
The definition of $\dual{\_}_{\tau}$ can be seen as extension of the construction in~\cite[ Corollary 1]{DEG22} for extending a functional of type $(\realLine\to \realLine)\to \realLine$ to $(\realDom\to \realDom)\to \realDom$.
\end{definition}
\begin{proposition} \label{starOpDescription}
For any second order type $\tau = \vec{\tau'} \to \delta$, functionals $f_1, f_2 \in \realLine_\tau$ and list of values $\overrightarrow{x_1}, \overrightarrow{x_2} \in \realLine_{\vec{\tau'}}$, by the infinitesimal property of $\varepsilon$, we have:
\begin{enumerate}[label=(\roman*)]
    \item $f_1 = \real{\dual{f_1}_\tau + \varepsilon \dual{f_2}_\tau}_\tau$
    \item $f_2 = \infi{\dual{f_1}_\tau + \varepsilon \dual{f_2}_\tau}_\tau$
    \item $\lDeriv f_1 \, \overrightarrow{x_1} \, \overrightarrow{x_2} = 
    \infi{\dual{f_1}_\tau \, (\overrightarrow{\dual{x_1}} + \epsilon \overrightarrow{\dual{x_2}})}_\delta$
\end{enumerate}
\end{proposition}
Note that for first-order types $\tau$, the functions $\dual{\_}_\tau$ are just set-theoretic functions since they are not continuous functions on the infinitesimal component.
In the above and in the following, we  use the point-wise extension of the multiplication by the dual number $\varepsilon$, and the addition operation $+$.
That is, if $*$ is an operation defined on the domain $D$, the operation $*$ on the domain $C \rightarrow D$, is defined by: 
\[
(f_1 * f_2)(c) = f_1(c) * f_2(c)
\]
and similarly for other operations and functions.

In the following, we will sometimes omit the type $\tau$ from $\dual{\_}_\tau$, $\real{\_}_\tau$, $\infi{\_}_\tau$ when the type of $\tau$ is clear from the context. We will also implicitly assume, where necessary, that any real number is automatically ``cast'' to a dual number.

Denote by $\St$ and $\In$ the envelopes of the functions  $\real{\_}_\delta$ and $\infi{\_}_\delta$ respectively.  They are functions in $\dualDom \rightarrow \realDom$ defined by:  $\St(x + \varepsilon x') = x$, $\In(x + \varepsilon x') = x'$.


\section{A language for differentiable functionals}\label{language}

Next we present Dual PCF, a language with a primitive operator for the evaluation of directional derivatives of functionals.  The language is a simply typed $\lambda$-calculus extended with a suitable set of constants. 


The types of the Dual PCF are defined by the grammar: 
\[ \tau \ ::=   \ o \  \mid \ \nu \ \mid \  \pi \ \mid\  \delta \ 
\mid \ \tau \rightarrow \tau
\] 
where $o$ is the type of booleans, $\nu$ is the type of natural numbers, $\pi$ is the type of real numbers,  and $\delta$ is the type of dual numbers.
The derivative operator is defined only on the type of dual numbers $\delta$. We assign to variable $x$ the type $\pi$ if we are not interested in evaluating the derivative with respect to $x$; values of type $\pi$ have implicitly an infinitesimal part equal to $0$.

The set of expressions in the language is defined by the grammar: 
\begin{equation} \label{grammar}
     e ::=  \lo{c} \ \mid\ x^\tau \ \mid\ e_1 e_2 \ \mid\ \lambda x^\tau . e 
\end{equation}
where $x^\tau$ ranges over a set of typed variables and $\lo{c}$ over a set of constants. 
For simplicity, here we present only a minimal set of basic constants, sufficient to express any other computable function.  In a real programming language this minimal set will be extended with other functions.
All constants defining functions on dual numbers, for example $\rmax : \delta \rightarrow \delta \rightarrow \delta$ have a corresponding version on real numbers $\rmax : \pi \rightarrow \pi \rightarrow \pi$, acting in the obvious way.  To avoid repetition, we present just the dual number versions, implicitly assuming the definition for the real number version. 



 
 The basic constants in the language are as follows:
\begin{enumerate}[label=(\roman*)]
    \item  The three total arithmetic operations,  $ +, -, *   : \delta \rightarrow \delta \rightarrow \delta$.
\item  Division by a positive natural number, $ /   : \delta \rightarrow \nu \rightarrow \delta$.
\item  Minimum and maximum $\rmin, \rmax  : \delta \rightarrow \delta \rightarrow \delta$, evaluating the minimum and maximum  of two dual numbers.
\item Two casting functions (explicit conversion), from naturals to reals $\oper{in_\pi} : \nu \to \pi$, and from reals to duals $\oper{in_\delta} : \pi \to \delta$.
\item A zero-test on reals  $(0 <) : \pi \rightarrow o$, that cannot be applied to dual values.
This restriction assures that functions on dual numbers do not have points of discontinuity on maximal elements. For example, a function, from dual values to dual values, returning $0$ on strictly negative values and $1$ on strictly positive ones, will not be definable. This fact, in turn, is necessary to guarantee the correctness of the derivative operator.
\item A projection function $\pr :  \delta \rightarrow \delta$, projecting a value on the unit interval $[-1,1]$. On real values, the behaviour of $\pr$ is defined by $\pr(x) = \rmax(-1, \rmin(x , 1))$.
\item An integration functional $\integral_\delta: (\pi \rightarrow \delta) \rightarrow \delta$, giving the Riemann integral of a function on the interval $[0,1]$. 

\item A supremum functional, $\sup: (\pi \rightarrow \delta) \rightarrow \delta$, evaluating the supremum of functions in $[0,1]$.
\item  {For second-order function types, $\tau = \vec{\tau} \to \delta$, a directional derivative operator, $\llderiv_{{\tau}}: (({\tau} \to \vec{\tau_\pi} \to \vec{\tau_\pi} \to \pi)$.  Each type $\tau_i$ in the list $\vec{\tau}$ must be equal to $\delta$, or be a first order type, in the form  $\sigma_1 \to \ldots \sigma_n \to \delta$, with $\sigma_1$ a ground type, }
while the type $\tau_\pi$ is recursively defined by $\delta_\pi = \pi$ and  $(\sigma_1 \to \ldots \sigma_n \to \delta)_\pi = \sigma_1 \to \ldots \sigma_n \to \pi$. Given a function $f$ in several arguments and returning a dual number, $\llderiv_{\vec{\tau}} f \vec{x} \vec{y}$ evaluates the derivative of $f$ at $\vec{x}$ along the direction $\vec{y}$. In the expression $\llderiv_{\vec{\tau}} f \vec{x} \vec{y}$, we need to add the syntactic condition that the derivative operator $\llderiv$ does not appear inside $f \vec{x} \vec{y}$.
\item The standard PCF constants on natural numbers, a sequential if-then-else test function and a fixed-point operator on arbitrary types $Y_\tau : (\tau \rightarrow \tau) \rightarrow \tau$.
\end{enumerate}

To ensure that for a function $f$ of type $\delta \to \delta$, for example, the infinitesimal part describes the derivative of the evaluated part, the language has restrictions on the way $\delta$ values are used. It is impossible to convert a dual into a real or to test whether a dual is less or greater than $0$. Consequently, all functions from dual numbers to Booleans are constant, so the if-then-else operator cannot be used to define functions on dual numbers that have no generalised derivative. As in \cite{DigEdalat13}, we have used the min/max operators as a safe alternative to if-then-else. 

\hypertarget{operational-semantics}{%
\subsection{Operational semantics}\label{operational-semantics}}
We define a small-step operational semantics. We note that we cannot use a standard PCF operational semantics as in \cite{MazzaPagani21} because in our work, unlike \cite{MazzaPagani21}, we implement exact computation on real numbers, so a real number cannot be defined by a single finite value. 
In \cite{DiGianan99,DigEdalat13}, an operational semantics to exact computation on real numbers is given by representing real numbers as streams of digits and using lazy evaluation to implement functions on them, but this has its drawbacks. It relies on parallel computation and is difficult to simulate in a programming language. In this paper, we propose an alternative approach, reminiscent of some work in \cite{bauer08,ShermanMC21}, with the advantage of a fairly direct translation into Haskell.

In the operational semantics, we use a set of basic constants for the rational intervals $[a,b]$, together with a set of basic constants on dual numbers, made up of pairs of rational intervals, $[a,b] + \varepsilon [a',b']$. We consider the infinite interval $(-\infty, + \infty)$ to be a special case of a rational interval. We avoid introducing these values directly in the main syntax of the language because they are partial values, and we prefer to have constants only for totally defined values, as is common in most programming languages. Using the functions $\oper{in_\pi} : \nu \to \pi$ and $\oper{in_\delta} : \pi \to \delta$ and the arithmetic operations, all rational values are readily available in the language.

In the operational semantics, we need to address the problem that by unfolding the fixed-point operator $Y_\tau$, on the one hand, one can obtain better and better approximations for a real value, but, on the other hand, the unfolding needs to be stopped at some point; otherwise, the computation diverges.  In exact real number computation, which deals with infinitary objects, one rarely has base cases in recursive definitions. For example, one can define a real number by the recursive equation $x = 1/4 + x/4$, and the infinite unfolding of this definition is an infinitary expression representing the value $1/3$.  The partial approximations of $1/3$ actually used in the computation are obtained by forcing a stop in the infinite unfolding. Similar considerations hold for the Riemann integral, by increasing the number of sub-intervals with which the unit interval is partitioned, one obtains better approximations of the integral, but the partition of the unit interval cannot be refined indefinitely.

To solve this problem, we extend the syntax, used in the operational semantics, by building expressions in the form $\langle e, n \rangle$ (or $\langle e, (m,n) \rangle$). The parameter $n$ (or $(m,n)$) represents a measure of the complexity of the computation along which the expression $e$ is going to be evaluated. Extending the syntax of the terms is a way to introduce in the syntax some information about the evaluation strategy that is useful in defining the reduction rules, but is not very meaningful in describing a function. One can avoid grammar extensions by adding extra information to the reduction rules, but the current approach is simpler to define.

The operational semantics allows to derive judgements in the form $\langle e, n \rangle \to_* [a,b] + \varepsilon [a',b']$, whose intended meaning is that with a computation bounded by a cost $n$, the expression $e$ reduces to the rational dual $[a,b] + \varepsilon [a',b']$.  
In more detail, Dual PCF has two forms of  the recursive operator $Y_\tau$: a bounded one, when $\tau$ is a \emph{continuous type}, a function type having a continuous range space, that is $\tau = \vec{\tau'} \to \pi$ or $\tau = \vec{\tau'} \to \delta$, and a standard one for any other value of type $\tau$.
Higher values of the parameter $n$ imply more effort in the computation so it will be always the case that if  $m \leq n$,  $(e, m) \to [a,b] + \varepsilon [a',b']$ and $(e, n) \to [c,d] + \varepsilon [c',d']$ then $[a,b] + \varepsilon [a',b'] \sqsubseteq [c,d] + \varepsilon [c',d']$. In other words, the evaluation of $\langle e, 0 \rangle, \langle e, 1 \rangle, \langle e, 2\rangle$, \ldots,  produces a sequence of intervals each one contained in the previous one and converging to the denotational semantics of $e$.
This approach to the operational semantics is somewhat similar to and inspired by~\cite{ShermanMC21,bauer08}. 

Formally for the operational semantics, we consider an extended language obtained by adding extra constants and three extra production rules to the expression grammar of Dual PCF, as in Equation~(\ref{grammar}): 
\[
    e::=  \langle e, n \rangle  \ \mid \ \langle \integral , (m,n) \rangle   \ \mid \ 
            \langle \osup , (m,n) \rangle
   \]
with $n, m \in \nat$.

The set of constants is extended by: 
\begin{itemize}
\item  a set of constants $[a,b] : \pi $, and $[a,b] + \varepsilon [a',b'] : \delta$, with $[a,b]$ and $[a',b']$ rational intervals, or the value $(-\infty,+\infty)$;
\item a constant $\oper{In}$ for a  function $\delta \rightarrow \pi$ returning the infinitesimal part of a dual number.
\end{itemize}


The evaluation contexts are the standard ones for a call-by-name reduction:
\begin{align*}
E[\_] ::= & [\_] \ \mid\ E[\_] \, e  \ \mid\ E[\_] \lo{op} e \ 
\mid\ c \lo{op} E[\_]  
\ \mid\ \lo{if} E[\_]  \lo{if} e_1  \lo{else} e_2
    \ \mid\ \oper{In} E[\_] 
\end{align*} 
The reduction rules for generic terms are: 
$$
\langle e_1 e_2, n \rangle \to  \langle e_1, n \rangle e_2 
\hspace{1.5em}
\langle \lambda x . e_1, m \rangle e_2 \to  \langle e_1 [e_2/x], m \rangle
$$
The reduction rules for a generic binary operation on dual numbers $\operatorname{op}$ are: 
$$
\langle e_1 \, \operatorname{op} \, e_2, n \rangle \to \langle e_1, n \rangle \, \operatorname{op} \, \langle  e_2, n \rangle 
$$
together with the rules defining operations on rational intervals; e.g.: 
\begin{align*}
\langle x + \varepsilon x' \, , \, n \rangle  & \to x + \varepsilon x'
\\
(x + \varepsilon x')  \mathbf{+} (y + \varepsilon y') & \to (x + y) + \varepsilon (x' + y')
\\
(x + \varepsilon x')   \mathbf{-}  (y + \varepsilon y') & \to (x - y) + \varepsilon (x' - y')
\\
(x + \varepsilon x') \mathbf{*} (y + \varepsilon y') & \to (x * y) + \varepsilon (x * y' + y * x')
\\
(x + \varepsilon x')  \mathbf{/} 0 & \to (-\infty,+\infty) + \varepsilon (-\infty,+\infty) 
\\
 (x + \varepsilon x')   \mathbf{/} n & \to (x/n + \varepsilon x'/n) 
\\
  \rmax (x + \varepsilon x') (y + \varepsilon y') & \to x + \varepsilon x' \mbox{ if }x > y
\\
  \rmax (x + \varepsilon x') (y + \varepsilon y') & \to y + \varepsilon y' \mbox{ if } y > x
\\
  \rmax ((-\infty,+\infty) + \varepsilon x') (y + \varepsilon y') & \to (-\infty,+\infty) + \varepsilon (x' \sqcap y')
\\
  \rmax (x + \varepsilon x') ((-\infty,+\infty) + \varepsilon y') & \to (-\infty,+\infty) + \varepsilon (x' \sqcap y')
\\
  \rmax (x + \varepsilon x') (y + \varepsilon y') & \to I + \varepsilon (x' \sqcap y') \mbox{ otherwise, } 
  \mbox{ with } I = [\max(x^-,y^-), & \max(x^+,y^+)]
  \\
  \pr (x + \varepsilon x') & \to (-1 + \varepsilon 0) \mbox{ if } x < -1
\\
  \pr (x + \varepsilon x') & \to (1 + \varepsilon 0) \mbox{ if } x > 1
  \\
  \pr (x + \varepsilon x') & \to (x + \varepsilon x') \mbox{ if } -1 < x < 1
\\
  \pr (x + \varepsilon x') & \to (x \sqcup [-1,1]) + \varepsilon (x' \sqcap 0) 
 \mbox{ otherwise }
 \end{align*}

The reduction rules for the functionals on dual numbers make use of the parameter $n$.  We have the following rules:
\begin{align*}
\langle \integral, n\rangle f  \to & \langle \integral, (n,n)\rangle f
\hspace{3em}
\langle \integral , (0,n)\rangle f \to \langle f, n \rangle ([0,1]) 
 \\
 \langle \integral (m,n)\rangle f   \to & 
    (\langle \integral, (m-1,n) \rangle  \lambda x \,.\, f( x / 2)) /2  
 + (\langle \integral, (m-1,n) \rangle \lambda x \,.\, f((x + 1)/ 2))/ 2
\temporaryRemoved{
\\
\langle \osup, n \rangle f    \to & \langle \osup, (n,n)\rangle f   
\\
\langle \osup, (0,n)\rangle f \to & \langle f, n\rangle ([0,1]) 
 \\
\langle \osup, (m,n)\rangle f  \to & \rmax 
(\langle \osup, (m-1,n)\rangle  \lambda x \,.\, f( x / 2)) \\
& (\langle \osup, (m-1,n)\rangle  \lambda x \,.\, f((x + 1)/ 2))}
\end{align*}
The operational semantics for the derivative operator is defined by:
\[
\langle \llderiv_{\vec{\tau}}, n \rangle  f \vec{e} \vec{e'} \to \oper{In} \langle f ({\vec{e} +_{\tau_i} (\cepsilon_\tau \vec{e'})}) , n \rangle 
\hspace{2em}
    \oper{In} ([a,b] + \varepsilon [a',b']) \to [a',b'] 
\]
where $+_\tau$ and $\cepsilon_\tau$ denote the expressions:  $+_{\sigma \to \tau} = \lambda f . \lambda g . \lambda x .  ( f x ) +_\tau (g x)$ and 
$\cepsilon_{\sigma \to \tau} = \lambda f . \lambda x . \cepsilon_\tau ( f x )$, while 
$\cepsilon (x + \varepsilon x')$ is a shorthand for  $(0 + \varepsilon) * (x + \varepsilon x')$.  

The operational semantics for the fixed-point $Y_\sigma$ operator on a continuous type $\sigma$ is defined by: 
\begin{align*}
\langle Y_{\sigma}, 0 \rangle f & \to \bot_\sigma  
\hspace{3em}
\langle Y_{\sigma}, (n + 1) \rangle f  \to \langle f (Y_{\sigma} f), n \rangle 
\end{align*}
  where $\bot_{\delta} = (-\infty, +\infty) + \varepsilon (-\infty, +\infty)$, 
  $\bot_{\pi} = (-\infty, +\infty), \bot_{\sigma \rightarrow \tau} = \lambda x_\sigma . \bot_\tau$.
{
The rules for the constants $+, -, *, /, \rmin, \rmax, \pr, \integral$, and $\osup$ acting on real values, are the obvious restriction of the corresponding rules for dual numbers.  

The reduction rules for the PCF constructors and constants are the standard ones, together with the rules:
\[ 
\oper{if} \langle e, n \rangle \oper{then} \langle e_1, n \rangle  \oper{ else } \langle e_2, n \rangle 
\langle c, n \rangle \to c
\] 
with $c$ any PCF constant, different from $Y_\sigma$, and with $\sigma$ a continuous type. }

The remaining rules can be found in the full version of the paper~\cite{DEG22}. 

\hypertarget{denotational-semantics}{%
\subsection{Denotational semantics}\label{denotational-semantics}}

The continuous Scott domain $\Dom_\tau$, used to give a semantic interpretation to expressions having arbitrary type $\tau$, is recursively defined by: $\Dom_o = \{ \lo{ff}, \lo{tt} \}_\bot$, $\Dom_\nu = \nat_\bot$, $\Dom_\pi = \realDom$, $\Dom_\delta = \dualDom$, $\Dom_{\sigma \to \tau} = (\Dom_\sigma \to \Dom_\tau)$. 


The semantic interpretation of any PCF constant is the usual one.  The general schema to give semantics to constants representing functions on dual numbers is the following: given a constant $\oper{c}$ of type $\tau$ that denotes a function $f_c$ on the real line ($\realLine$), the semantic interpretation of $\oper{c}$ is given by $\bsem{ c }$ defined by: 
\[ \bsem{ c } = \dual{ f_c }_\tau \]
To help the reader, we explicitly define the semantic interpretation of some of these constants:
\begin{align*}
 &  \bsem{\rmin }(x + \varepsilon x') (y + \varepsilon y') = 
  (\min\,x\,y) + \varepsilon
  \begin{cases} 
    x' & x < y \\ 
    y' & y < x \\ 
    x' \sqcap y' & \text{o/w} 
    \end{cases} 
\\
 \temporaryRemoved{
 &  \bsem{\osup} f =  [a,b] + \varepsilon \bigsqcap\left\{\, \In(f(x)) \, \mid \, \right.
 x \in [0,1],  \St(f(x)) \consist [a,b]\left.\right \}  \\
 &  \hspace{10ex}\mbox{ with } a = \sup_{x \in [0,1]} (\St(f(x)))^-  \mbox{ and } b = \sup_{x \in [0,1]} (\St(f(x)))^+  \\ }
 &  \bsem{\integral } f =  \int^\mathsf{d}_{[0,1]} f(x) dx 
\end{align*} 
Integration on the dual domain
is reduced as
\[\int^\mathsf{d}_{[0,1]} f(x) dx  = \int_{[0,1]}^\star \St(f(x)) dx + \varepsilon \int_{[0,1]}^\star \In(f(x)) dx,\] 
where $\int^\star_{[0,1]}:(\interval[0,1]\to\realDom )\to \realDom$ is the envelope of the Riemann integral functional $\int_{[0,1]}:([0,1]\to \realLine)\to\realLine$ as in Proposition~\ref{envelope}; it coincides with integration constructor developed in Real PCF and in interval analysis~\cite{edalat2000integration},~\cite{moore1966interval}. It sends a continuous function of type $\interval[0,1]\to \realDom$ to an interval in the domain of reals $\realDom$ and extends the Riemann integral in the sense that if $f:[0,1]\to \realLine$ is continuous then $\int_{[0,1]} f^\star(x)\,dx=\int_0^1 f(x)\,dx$, where as usual we identify a real number with its singleton. Note that the integration constructor, using the above method, can compute the value of $\int_{[a,b]}f(x)\,dx$ by a simple change of variable invoked by the linear re-scaling $x=h(y)=(b-a)y+a$ so that 
\begin{equation} \label{change-var}
   \int_{[a,b]}f(x)\,dx=(b-a)\int_{[0,1]}f\circ h^\star(y)\,dy. 
\end{equation}

For clarity, to distinguish the classical Riemann integral from the domain-theoretic integral, we always denote the classical Riemann integral of $f:[0,1]\to \realLine$ over any interval $[a,b]$ by $\int_a^b f(x)\,dx$ while $\int_{[u,v]}g(x)\,dx$, i.e., with the range of integration as a subscript to the integral sign, always denotes the extended interval-valued Riemann integral for a continuous function $g:\realDom\to \realDom$. 

The semantic interpretation of the derivative operator $\llderiv_{\vec{\tau}}$  is defined by:  
\[\bsem{\llderiv_{\vec{\tau}}} f \vec{d} \vec{e} = \In(f (\vec{d} + \varepsilon \vec{e})) \]
The interpretation of the other constants that cannot be obtained by the general scheme is the following: 
\begin{align*}
&  \bsem{\oper{in_\delta}} x = x + \varepsilon 0 \\
&  \bsem{(0 <)} x = \begin{cases} \ttt,& x > 0 \\ \ff,& x < 0 \\ \bot,& \text{otherwise} \end{cases}
&  \bsem{\oper{Y}_\sigma } f  = \bigsqcup_{i \in \nat} f^i (\bot_\sigma) 
\end{align*}
We point out that quite often dual numbers or functions on dual numbers obtained by using the recursion operator $\oper{Y}$ have an unbounded infinitesimal part, meaning that, depending on the type, the infinitesimal part is $\bot = (-\infty, +\infty)$, or is the function that maps every element to $\bot$. 
A simple example is the following recursive definition of the value $1$:
$\oper{Y} (\lambda x^\delta \,.\, (\pr x^\delta + 1)/2)$.
This fact can be explained as follows: the semantic interpretation functions are linear on the infinitesimal parts, and a linear function when applied to the bottom value $(-\infty,+\infty)$ returns either $0$ (if the linear map is identically $0$) or the bottom value itself. It follows that each element in the chain $(f^i (\bot_\sigma))_{i\in \nat}$, whose least upper bound gives the semantics interpretation of $Y f$, has an unbounded infinitesimal part.  

A solution for this problem, as in~\cite{DigEdalat13}, consists of introducing a second type of dual values, $\delta^l$, having the infinitesimal part bounded by the interval $[-1,1]$.  The basis functions on the type $\delta^l$ need to be non-expansive. Therefore most of the basic functions defined on $\delta$ must be replaced by non-expansive versions of them. For example, addition $+$ is replaced by a function evaluating the average of two values. To motivate this restriction, notice that by using addition it is possible to build a function that doubles its argument: $\lambda x . x + x$; this function maps $0 + \varepsilon 1$ to $0 + \varepsilon 2$, and therefore cannot have type $\delta^l \to \delta^l$.   Inside the type $\delta^l$, it is possible to use the fixed point operator to obtain functions with informative infinitesimal parts. The functions obtained in this way can later be embedded in the larger type $\delta$.  For lack of space and to focus on the main subject of this paper, which is evaluating derivatives of second-order functionals, we do not fully present this solution and refer the interested reader to \cite{DigEdalat13}.  


The semantic interpretation function \(\mathcal E\) is defined, by structural induction, in the standard way:
\[\begin{aligned}
  \esem{c}_\rho = \bsem{c} & \hspace{2em}
  \esem{x}_\rho = \rho(x)\\
  \esem{e_1 e_2}_\rho = \esem{e_1}_\rho(\esem{e_2}_\rho)  & \hspace{2em}
  \esem{\lambda x^\sigma.e}_\rho = \lambda d \in \mathcal D_\sigma.\esem{e}_{(\rho[d/x])}\\
\end{aligned}\]

\subsection{Adequacy}\label{adequacy}

The correspondence between the denotational and the operational semantics is shown by the following result. 
For a closed expression $e$ of type $\delta$,
let us denote by \([a,b] + \varepsilon[a',b'] \ll Eval(e)\) the property that there exists a natural number $n$ and a dual rational interval \([c, d, ] + \varepsilon[c',d']\) such that \( \langle e, n \rangle \to^* [c, d] + \varepsilon[c',d']\) and 
$[a,b] + \varepsilon [a',b'] \ll [c,d] + \varepsilon [c',d']$.

\begin{theorem}~\label{sound-complete}
On type $\delta$ the operational semantics is sound and complete with respect to the denotational semantics, that is for any 
closed expression $e : \delta$, for any partial rational dual number $[a,b] + \varepsilon [a',b']$, we have:
$[a,b] + \varepsilon [a',b'] \ll \esem{e}$ iff $[a,b] + \varepsilon [a',b'] \ll Eval(e)$.
\end{theorem}


\section{Some functions and functionals in Dual PCF} \label{examples}

In this section, we will give various examples of functions and functionals expressible in our language. In some cases, we will also show how to use the operational semantics to compute the derivatives of these functionals. Some of these examples are motivated by actual areas of application. In many cases (for instance, \cite{chebop}), the problem of solving an integral or differential equation can be reduced to finding the roots of an integral or differential operator; being able to evaluate the derivative of a functional is required for such problems:

\subsection{Absolute value function}
The absolute value function $f(x) = |x|$ can be written in the language as $\lambda x. \,\rmax(x,-x)$. The domain-theoretic directional derivative of this function at $0$ is then correctly evaluated by the following reduction:
    $$\begin{aligned}
    \llderiv_{\delta} f\, 0\, 1 &\to \In(f (0 + \varepsilon 1))
    \to \In(\max(0 + \varepsilon 1, 0 - \varepsilon 1)) 
    &\to \In(0 + [-1,1] \varepsilon) 
    \to [-1,1]
    \end{aligned}$$
\subsection{Comparison with Chebyshef software} From~\cite[section 3.2]{frechet}, consider taking the directional derivative of the operator $G = \lambda g.\, \lambda x. \, x + g(x)^2$ at the point $f = \lambda u. \, u^2$ in the direction $k$. We refer to the operational semantics:$$\begin{aligned}
    & \llderiv_{\delta \to \delta, \, \delta} G (\lambda u. \, u^2)\, y\, k\, 0 
    \to \In(G (\lambda u.\, u^2 + \varepsilon  k(u)) (y + \varepsilon 0))\\ 
    &\to \In(\lambda x .\, x + (x^2 + \varepsilon  k(x)) * (x^2 + \varepsilon k(x))) (y + \varepsilon 0)) \\
    &\to \In(y + (y^4 + \varepsilon 2  y^2 k(y))) 
    \to 2 y^2  k(y)
    \end{aligned}$$
    So in other words, we have that $LG(f,k) = \lambda y. 2y^2 k(y)$. This is the same result as was obtained by the software system in the above paper, except that their autodiff procedure is far more involved than ours.
    \subsection{Laplace and Fourier series} We can express in the language all continuous linear functionals like the Fourier-series finding operator and the Laplace transform of functions with compact support; see Section~\ref{lin-func}. The dual-number formalism implies that given a continuous linear functional $F$ and function $f$, we have: $F(f + \cepsilon g) = F(f) + \cepsilon F(g)$ for any function $g$. So we can deduce: $\lDeriv F(f,g) = F(g)$.
    \subsection{Lagrangian} Consider a functional of the form $F(f) = \int_0^1 {\mathfrak {L}}(t,f(t),f'(t))\,dt$ where ${\mathfrak{ L}}: \mathbb R^3 \to \mathbb R$, called the Lagrangian, and $f : \mathbb R \to \mathbb R$ are differentiable. This is what is called an action functional that is used for finding the time evolution of a system in analytical mechanics and quantum mechanics (\cite{kibble2004classical,sicm,weinberg2015lectures}).
    Assuming a term $\mathfrak{L} : \delta \rightarrow \delta \rightarrow \delta \rightarrow \delta$,  in our language $F$ is expressed as  $\lambda f \,.\, \integral (\lambda t \, . \, {\mathfrak {L}} (\oper{in_\delta} \, t) (f (\oper{in}_\delta \, t))(\llderiv_\delta f \, t \, 1))$.
     Our language can not be used to differentiate such a functional directly (because we cannot differentiate expressions that themselves contain the derivative operator). Instead, consider $G(g) = F\left(\int g\right) = \int_0^1 {\mathfrak {L}}(t, \int_0^t g(u)\, du, g(t))\, dt$. 
    The functional $G$ can be expressed in our language as \\
    $\lambda g \,.\, \integral (\lambda t \, . \, {\mathfrak L} \, (\oper{in}_\delta \, t) \,  (\integral (\lambda u \,.\, (\oper{in_\delta} \, t) \ast (g (\oper{in_\delta} (u * t))))) \, (g (\oper{in}_\delta \, t)))$, and can be differentiated (along $g$ and along the variables appearing in $\mathfrak{L}$). 
    
    Whenever the functional $G$ is locally stationary at $g$, the functional $F$ is locally stationary at the function $x \mapsto \int_0^x g(t)\,dt$. We can also fix the value of $\int_0^1 g(t)\,dt$ using Lagrange multipliers, to ensure that the starting and ending points of a particle are fixed, as is usually required when applying the principle of stationary action.
    \subsection{Solving initial value problems (IVP).}
Consider the 1d problem,
\begin{equation}\label{ivp} \dot{y'}(x) = v(y(x)),\qquad y(0)=0  
\end{equation}
where $v:O\to \realLine$ is a continuous one dimensional vector field in an open neighbourhood $O\subset \realLine$ of the origin. In our language the solution of the initial value problem can be expressed in the following way. Let $e_v : \pi \rightarrow \pi$, be an expression defining the function $v$ in its domain of definition. The solution of IVP is given by the expression:
\[Y (\lambda f \,.\,  \lambda x \, . \, \integral \lambda t \, . \,  x * \operatorname{pr}_M ( e_v ( f (t * x))))\]
where $\operatorname{pr}_M$ is a function projecting the real line onto the interval $[-M,M]$ and defined as $\lambda x . (\oper{in}_\pi M) * \pr x / M$, with $M$ a rational constant. The proof of correctness is given in~\cite{DEG22}. By using currying, the above construction can be easily extended to solve the IVP in $\realLine^n$.
\subsection{Legendre-Fenchel transform} Consider a function $f:[0,1] \to \mathbb R$. 
    Extend it to real values outside $[0,1]$ by setting $f(x) = +\infty$ whenever $x \in \mathbb R \setminus [0,1]$. We may define the \emph{Legendre-Fenchel transform} (\cite{convex}) of such functions as $F(f)(p) = \sup_{x \in [0,1]}(px - f(x))$. The resulting functional is definable in our language and is differentiable in the generalised sense we consider.
    \subsection{Thomas-Fermi kinetic energy functional} The Thomas-Fermi kinetic energy functional¬(\cite{kineticenergy}) is a classic example of a non-linear functional in quantum theory. It expresses in an approximate way the total kinetic energy of the electrons in a region of space purely in terms of their density distribution $n$. It is given by $$T(n)=C_{\rm kin}\int [n(\mathbf{r})]^{5/3}\ d^3r,$$ where $C_{\rm kin}$ is a constant involving the mass of an electron.
    If we suppose that $n : \mathbb R^3 \to \mathbb R$ is zero outside $[0,1]^3$, we can express the above functional in our language as:
    $$T(n) = C_{\rm kin} \int_0^1 \int_0^1 \int_0^1 (n(x,y,z))^{5/3}\,dx\,dy\,dz.$$

\section{Local consistency in the dual domain}\label{consistency-section}

To use dual numbers in a useful way, one needs to ensure that the infinitesimal value can actually be used to evaluate the derivative of a function.  In other words, one needs to prove that automatic differentiation provides correct results.  The problem of correctness of automatic differentiation is often considered in the literature; see for example~\cite{abadiPlotkin19} and~\cite{MazzaPagani21}. In our setting, where we admit also partial values, we require that the information about a function contained in the value part is not in contradiction with that of the derivative contained in the infinitesimal part.  We call this notion \emph{local consistency}.   

\begin{definition}
An element $x + \varepsilon x' : \dualDom$ is \emph{locally consistent} if $x'$ is an interval containing $0$. Given a function type $\vec{\tau} \to \delta$ with dual numbers $\delta$ as its codomain, 
a continuous functional $F : \Dom_{\vec{\tau} \to \delta}$ is  \emph{locally consistent} if it satisfies:
\begin{itemize}
    \item[(i)] it is \emph{standard-robust}, i.e., for any pair of lists  of standard robust values $\vec{d}, \vec{d'}$, 
\[
  \St(\vec{d}) = \St(\vec{d'}) \ \Rightarrow  \ \St(F(\vec{d})) = \St(F(\vec{d'})),
\]
\item[(ii)] for any pair of lists of  standard maximal preserving, locally consistent values $\vec{d}, \vec{d'}$, and rational $r>0$, 
\[
     \In(F({(\vec{d} \meet (\vec{d} + r \vec{d'})) + \varepsilon \vec{d'}}))
     \consist  \frac{\St (F(\vec{d} + r \vec{d'}) - F(\vec{d}))}{r}.
\]\end{itemize}
\end{definition}

\begin{lemma}~\label{loc.bound.f}
Suppose $F : \Dom_{\vec{\tau} \to \delta}$ is a standard-robust and standard maximal preserving continuous functional. If $\lDeriv\real{F}(\real{\vec{d}}, \real{\vec{d'}})$ is bounded for some pair $({\vec{d}},{\vec{d'}})$ of lists of locally consistent and standard preserving elements, then the restriction of $\real{F}$ to the space of maximal preserving elements is locally Lipschitzian at $(\real{\vec{d}},\real{\vec{d'}})$. 
\end{lemma}

The following result states that local consistency exactly characterizes the set of functions on $\dualDom$ for which the infinitesimal part gives a correct approximation of the directional derivative of the function represented by the standard part.

\begin{proposition} \label{localConsistencyDerivative2}
A standard-robust and standard maximal preserving continuous functional $F : \Dom_{\vec{\tau} \to \delta}$ 
is  locally consistent if and only if for any pair of lists of locally consistent and standard maximal preserving elements $\vec{d}$ and $ \vec{d'}$, we have:
\begin{equation}\label{directional-approx} \In(F(\vec{d} + \varepsilon \vec{d'})) \sqsubseteq 
\lDeriv \real{F}(\real{\vec{d}}, \real{\vec{d'}}). \end{equation}


\end{proposition}
If the left hand side of Equation~(\ref{directional-approx}), evaluates to a non-bottom element then by~\cite[Lemma 1]{DEG22}, $f$ is locally Lipschitzian and thus by \cite[Theorem 1(i)]{DEG22} implies that the left hand side of Equation~(\ref{directional-approx}) computes approximations to the support function of the generalised subgradient of the standard part of $F$. 
\section{Correctness of the derivative operator}\label{correctness}

By Proposition~\ref{localConsistencyDerivative2}, on locally consistent functions,  one can evaluate the derivative of a function by evaluating the function on dual numbers.  We aim now to prove that the semantic interpretations of the functions definable in the language are locally consistent.  

To obtain these results we need to extend the notion of consistency to all functions definable in our language, and in particular to functions of higher order than 2, and then prove that local consistency is preserved by the basic constructors of the language: these being function composition, $\lambda$-abstraction, lub of chains.  In order to obtain these results it is convenient to introduce a second, equivalent definition of consistency that is more suitable for extension to higher types.  Note that with the present definition of local consistency, the simple statement that composition of two locally consistent functions is locally consistent does not have a straightforward proof.  The second definition of consistency is based on logical relations. 
 Logical relations are a standard proof technique used in the semantics of functional languages; they are used for proving that the semantic interpretation of terms satisfies some desired properties. A general introduction to logical relations can be found in \cite{Mit96}, while in \cite{BartheCDG20,DigEdalat13}, logical relations are used in a way similar to our work. 
 
We define a set of logical relations which, if they are preserved by a function $f$, imply that the function $f$ is locally consistent. For any rational number $r>0$ let $R_{\delta}^{r}$ be a ternary relation over generalised dual numbers $\dualDom$ defined by:  $R_{\delta}^{r}(x_1, x_2, x_3)$ holds whenever 
 \[\St(x_3) \sqsubseteq \St(x_1) \sqcap \St(x_2) \mbox{ and } \In(x_3) \consist \frac{\St (x_2 - x_1)}{r}\] 
 or, equivalently, 
 $R_{\delta}^{r}(I_1 + \varepsilon I_1', I_2 + \varepsilon I_2', I_3 + \varepsilon I_3')$ holds whenever $I_3 \sqsubseteq I_1 \sqcap I_2$ and $r I'_3 \consist {I_2 - I_1}$.
 
For the other ground domains, $\Dom_o$ and $\Dom_\nu$, $R_{o}^{r}$ and $R_{\nu}^{r}$ are defined as follows: 
$R_{\nu}^{r}(n_1, n_2, n_3)$ holds whenever $n_3 \sqsubseteq n_1 \sqcap n_2$, and $n_1, n_2$ are consistent. 
The rationale behind this definition consists in repeating the definition of $R_{\delta}^{r}$  by considering Boolean values  and natural numbers as having a hidden infinitesimal part equal to zero.

The relations  are extended inductively to higher order domains in the usual way for logical relations:  $R_{\sigma \to \tau}^{r}(f_1, f_2, f_3)$ iff for every $d_1, d_2, d_3 \in \Dom_\sigma$, the relation $R_{\sigma}^{r}(d_1, d_2, d_3)$ implies $R_{\tau}^{r}(f_1(d_1), f_2(d_2), f_3(d_3))$.

\begin{definition}
An element $f$ in the domain $\Dom_\sigma$ is \emph{logically consistent} if it is self-related by $R_{\sigma}^{r}$, i.e. $R_{\sigma}^{r}(f, f, f)$, for any positive rational number $r$. We call a constant $c$ in the language \emph{logically consistent} if its semantic interpretation $\bsem{c}$ is logically consistent.
\end{definition}
 \begin{proposition}~\label{equivalenceConsistency1}
Any first-order function $f : \Dom_\tau$  is locally consistent if and only if it is logically consistent.
\end{proposition}
Next, we have the following implication regarding functionals. We conjecture that the reverse implication also holds, but since the reverse implication is not required in this work, we avoid considering it here.  

\begin{proposition}\label{functional-cons}
Any second-order function $F : D_{\vec{\tau} \to \delta}$ is locally consistent if it is logically consistent.
\end{proposition}

Note that, with the single exception of $\llderiv_{\vec{\tau}}$, all the constants in the language are logically consistent. The proof is routine for almost all  constants.  To prove that the fixed-point operator preserves the above relations, one shows that the bottom elements are self-related by $R_\sigma^{r}$, and that the relation is closed under the lub of chains. Note that $(<0)$ preserves the relation when applied to the domain $\Dom_{\pi}$, but it has no logically consistent extension to the domain $\Dom_{\delta}$.

Using the technique of logical relations \cite{Mit96}, it is straightforward to show:

\begin{proposition}
 The semantic interpretation $\esem{e}$ of any closed expression $e: \tau$ not containing $\llderiv$ is logically consistent.
\end{proposition}

\begin{corollary}
The semantic interpretation $\esem{e}$ of any closed expression $e$ having second-order function type and not containing $\llderiv$ is locally consistent.
\end{corollary}



\begin{corollary}
The derivative operator $\llderiv$ is sound, i.e., for any closed expression $\llderiv F (\vec{f})(\vec{g})$, if $F$ is a second-order function, $\esem{F}, \esem{\vec{f}}$ are standard maximal preserving, and $\llderiv$ is not contained in $F, \vec{f}, \vec{g}$ then 
\[
\esem{\llderiv F (\vec{f})(\vec{g})} \sqsubseteq \lDeriv \real{\esem{F}} (\real{\esem{\vec{f}}}, \real{\esem{\vec{g}}})
\]
\end{corollary}

Since our language contains the if-then-else operator, and it is a well-known problem that the if-then else constructor produces an inconsistent result with automatic differentiation, the above result may appear contradictory. Note, however, that there are restrictions on the functions that can be defined on dual numbers. In particular, it is impossible to convert a dual number into a real number, or to test whether a dual value is less or greater than 0. Consequently, all functions from duals to Booleans are constant, so the if-then-else operator cannot be used to define functions on duals that have no generalised derivative. As in \cite{DigEdalat13}, we have used the min/max operators as a safe alternative to if-then-else.

\section{Definability of Linear functionals}\label{lin-func}
We show that any computable continuous linear functional on $C_0(\realLine)$, the set of continuous real functions that vanish at infinity, can be expressed in our language. We start with continuous linear functionals on $([0,1]\to \realLine)$.
\begin{proposition}\label{R-rep}
If $F:([0,1]\to {\mathbb R})\to {\mathbb R}$ is a continuous linear functional then there exist two right continuous non-decreasing maps $g_i^\dagger:[c_i,d_i]\to[0,1]$ with $i=1,2$ such that for any continuous map $f:[0,1]\to {\mathbb R}$ we have 
\[F(f)=\int_{c_1}^{d_1} f\circ g^\dagger_1(x)\,dx -\int_{c_2}^{d_2} f\circ g^\dagger_2(x)\,dx.\]
\end{proposition}
By employing the envelopes of $g^\dagger_1$ and $g^\dagger_2$ as in Proposition~\ref{envelope}, we can then deduce the following result. 
\begin{theorem}~\label{clf-int} The continuous linear functional $F:([0,1]\to {\mathbb R})\to{\mathbb R}$ has a conservative extension ${F}^\star:(\interval[0,1]\to \interval {\mathbb R})\to \interval{\mathbb R}$ given by
\[{F^\star}(f)=\int_{[c_1,d_1]}f\circ (g_1^\dagger)^\star(x)\,dx-\int_{[c_2,d_2]}f\circ (g_2^\dagger)^\star(x)\,dx.\]
Similarly, the embedding $\dual{F}$ of $F$ in the dual domain $\Dom_{(\delta \to \delta) \to \delta}$ is given by:
\[\dual{F}(f)=\int^{{\mathbf{d}}}_{[c_1,d_1]}f \dual{(g_1^\dagger)^\star \real{x}}dx - \int^{\mathbf{d}}_{[c_2,d_2]}f \dual{(g_2^\dagger)^\star \real{x}}dx.\]
\end{theorem}
Theorem~\ref{clf-int} reduces the definability of a computable linear functional $F$ to the definability of computable functions $(g_1^\dagger)^\star$ and $(g_2^\dagger)^\star$. The latter result is proved in \cite{DiGianan99,Esc96}, for a calculus whose expressing power, on first order functions, is equivalent to the present calculus. Thus, in our calculus computable linear functionals of type $([0,1]\to\realLine)\to\realLine$ are definable. Next, consider $C_0(\realLine)$. We show in~\cite{DEG22}, using Riesz-Markov representation theorem and Hahn decomposition of finite signed measures, that Theorem~\ref{clf-int} can be extended to continuous linear functionals on $C_0(\realLine)$. Therefore, any computable continuous linear functional on $C_0(\realLine)$ can be expressed in our language. In particular, this holds for any continuous linear functional on $C_0(\realLine)$ equipped with the compact-open topology, since the compact-open topology is weaker that the sup norm topology on $C_0(\realLine)$. 

\section{Conclusions and future work}
We have for the first time developed a language to compute the derivatives of functionals on real-valued functions. Dual PCF contains, as basic primitives, the operations of integration and sup. Using this language it is possible to numerically evaluate the directional derivative  of functionals. We have also reduced the problem of definability of continuous linear functionals to the previously solved problem of definability of real functions. 

The present work can be expanded in different directions. As an immediate application, one can properly implement the primitives of Dual PCF in a programming language.  This will provide a simple and new method to numerically evaluate the derivative of functionals.  In this paper, to ensure correctness of computations, and to avoid dealing with errors induced by floating point arithmetic, we assume that we are computing with exact real numbers. However, there would be no problems in defining our primitive operations over floating point numbers instead.  One just needs to encode the infinitesimal part of a dual number, representing the directional derivative of a function, by an interval.

We conjecture that using supremum and integration all computable functionals, and not just the computable linear functionals,  would be definable.  

We may extend our language to support nested differentiation. A common approach is to introduce a set of ``independent'' infinitesimals, each like $\varepsilon$ in that they square to zero, but such that their mutual products are not zero \cite{ManzyukPRRS19}.





\bibliography{dualNumbersMFPS}
\bibliographystyle{entics}
\end{document}